\documentclass[a4paper,11pt]{article}
\usepackage{jheppub} 
\usepackage{lineno}
\usepackage{amsmath, mathtools, amssymb, hyperref, xcolor, amsthm, pdfpages, tensor, tikz-cd}

\newcommand{\extp}{\textstyle{\bigwedge^{\!\bullet}}}

\newcommand{\diff}{\mathrm{d}}
\newcommand{\pdiff}[2]{\frac{\partial #1}{\partial #2}}
\newcommand{\tr}{\operatorname{tr}}

\newcommand{\ddt}{\left.\frac{\diff}{\diff t}\right\vert_{t=0}}

\newcommand{\trg}{\diamondsuit}
\newcommand{\gh}{{\mathrm{gh}}}

\DeclareMathOperator{\vol}{vol}

\newtheorem{theorem}{Theorem}
\newtheorem{lemma}{Lemma}
\newtheorem{corollary}{Corollary}
\theoremstyle{definition}

\title{Einstein-Cartan and Palatini-Cartan Gravity from Graded Poisson Geometry}

\author{Johannes Moerland}
\emailAdd{jmoerland@constructor.university}
\author{and Peter Schupp}
\emailAdd{pschupp@constructor.university}
\affiliation{Constructor University Bremen,\\
Campus Ring 1, 28759 Bremen, Germany}

\abstract{Using notions from graded Poisson geometry, we develop a framework that allows for concise formulations of field theories such as Einstein-Cartan and Palatini-Cartan gravity. We show that Einstein-Cartan gravity is equivalent to a variant of Palatini-Cartan gravity in which the vielbein fields are constrained to satisfy appropriate symmetry relations, which we refer to as \emph{constrained vielbein} gravity. To that end, we  relate these theories via a diffeomorphism of graded manifolds. We then show how to obtain the full Palatini-Cartan theory by removing the symmetry constraint.
Finally, we show how our graded geometric framework allows for a direct incorporation of the minimal BV extension of the graded constrained vielbein and Palatini-Cartan theories.}

\begin{document}
\maketitle
\flushbottom

\section{Introduction}
First-order gravitational field theories cast the metric structure of spacetime and the connection as a priori independent fields. As such, the Einstein-Cartan (EC) and Palatini-Cartan (PC) theories provide suitable frameworks for describing spinning particles in a gravitational field~\cite{hehl1976general}. While the EC theory is classically formulated in terms of a metric $g_{\mu\nu}$, the PC theory is conventionally built around a vielbein field $\tensor{e}{_\mu^i}$. The theories can then be related by $g_{\mu\nu} = \tensor{e}{_\mu^i} \eta_{ij} \tensor{e}{_\nu^j}$, where $\eta_{ij}$ is the Minkowski metric.

It is classically known that the EC and PC theories lead to equivalent field equations. Since vielbein fields generally contain more information than the metric, one should not expect the full variational principles to be equivalent. In this paper, we show that the variational problem of EC gravity is equivalent to a variant of PC gravity where vielbein fields satisfy appropriate symmetry constraints. We shall call this theory \emph{constrained vielbein (CV)} gravity. We then show how upon removing these constraints, one recovers the full PC theory.

To do so, it proves useful to cast the field theories in the language of graded geometry.
It is well-known that graded geometry efficiently describes physical theories in the case where the underlying spacetime manifold is Minkowski space~\cite{medeiros2003exotic}, only local expressions are considered~\cite{chatzistavrakidis2017tensor,chatzistavrakidis2020graded,chatzistavrakidis2020unified}, or that the manifold is equipped with a fixed metric and the associated Levi-Civita connection~\cite{bruce2018graded, deser2018extended}. However, techniques for dealing with dynamical metrics and connections are less developed. We demonstrate how graded Poisson geometry~\cite{roytenberg2002structure, cattaneo2006graded,basile2025supersymmetric} provides a natural framework for addressing these dynamical geometric structures. To that end, we give an intrinsic characterization of the relevant graded Poisson structures on graded manifolds. The resulting graded geometric formulation is a mixed symmetry generalization of differential forms and works on arbitrary curved manifolds. As such, it is as powerful in tensor theories including gravity as differential forms are in gauge theories. The geometric structures streamline many computations and directly allow for an incorporation of the minimal BV extension of the PC theory.

After establishing the relevant structures, we formulate the graded EC theory on $\mathcal M = T[1]M \oplus T[1]M$. The action functional takes the form
\begin{equation}
    S_\textrm{EC}(g,D) = \frac{2}{(n-2)!} \int_\mathcal M g^{n-2} R_D \mu_g,
\end{equation}
where $g$ is a function on $\mathcal M$ that encodes a Poisson tensor, $\mu_g$ is a Berezin measure, $D$ is a Poisson derivation, and $R_D$ measures the failure of $D^2 = 0$.

The graded CV and PC theories, on the other hand, are formulated on the graded manifold $\mathcal N = T[1]M \oplus T^*[1]M$. The former may be obtained by pulling back the EC theory on $\mathcal M$ along a suitable diffeomorphism of graded manifolds $\Phi:\mathcal N\to\mathcal M$ constructed from the Poisson structure. After a field redefinition, the resulting action functional can then be written as
\begin{equation}
    S_\textrm{cv}(e,D) = \frac{1}{(n-2)!}\int_\mathcal N e^{n-2} R_D [dx\,d\theta\,d\chi],
\end{equation}
where $e$ is a constrained function on $\mathcal N$ and encodes a vielbein field, $D$ is again a Poisson derivation with the failure of $D^2=0$ captured by $R_D$, and $[dx\,d\theta\,d\chi]$ denotes the flat Berezin measure induced by a choice of coordinates. 

Upon removing the constraint on $e$, one recovers the graded PC theory. Dropping the flat Berezin measure from notation, the action functional reads
\begin{equation}
    S_\textrm{PC}(e,D) = \frac{1}{(n-2)!}\int_\mathcal N e^{n-2} R_D,
\end{equation}
where $e$ is now a function on $\mathcal N$ corresponding to an \emph{unconstrained} vielbein field, $D$ is again a compatible Poisson derivation, and $R_D$ quantifies the failure of $D^2=0$.

Since the gravitational theories we consider are invariant under diffeomorphisms of $M$, we construct the appropriate Batalin-Vilkovisky extension of graded PC action.

This paper is structured as follows: We first review some basic notions from graded geometry in section~\eqref{sec:graded-geometry}. In section~\eqref{sec:metric-gravity}, we specialize the general notions to $\mathcal M$ and construct the required graded geometric structures on curved space-time. Using these tools, we construct the graded Einstein-Cartan action functional. In section~\eqref{sec:vielbein-gravity}, we use these results to derive the constrained vielbein theory and show how to obtain the Palatini-Cartan theory. Finally, we construct the minimal BV extension of the graded CV theory.

\section{Graded Geometry}\label{sec:graded-geometry}
In this section, we introduce the fundamental concepts and tools that will be used throughout the remainder of this paper. We give a short conceptual introduction fitted to our needs, more detailed constructions can be found in, e.g.,~\cite{batchelor1984graded, fairon2017introduction}. For each vector bundle $E\to M$ of rank $r$ and each odd integer $k$, let us consider the graded manifold $T[1]M \times_M E[k]$, where $\times_M$ denotes the fiber product over $M$. By minor abuse of notation, this graded manifold is more commonly written as $T[1]M \oplus E[k]$, which we shall denote by $\mathcal E$ in the following. Locally, we can equip $\mathcal E$ with coordinates $(x^\mu,\theta^\nu,\xi^i)$, where $x^\mu$ are the degree 0 coordinates on $M$, $\theta^\nu$ correspond to the directions originating from $T[1]M$, and $\xi^i$ correspond to the directions coming from $E[k]$. The algebra of functions is locally isomorphic to
\begin{equation}
\mathcal O_\mathcal E (U) = C^\infty(U)[\theta^1,\dots,\theta^n; \xi^1,\dots,\xi^r] = C^\infty(U) \otimes \extp (\mathbb R^n)^* \otimes\extp(\mathbb R^r)^*,
\end{equation}
where $U\subset M$ is an open subset. In particular, all graded generators are taken to mutually commute in the graded sense
\begin{equation}
\theta^\mu\theta^\nu + \theta^\nu\theta^\mu = 0,\quad \theta^\mu\xi^i + \xi^i\theta^\mu = 0,\quad  \xi^i\xi^j + \xi^j\xi^i = 0.
\end{equation}
A general graded function may be expanded in graded coordinates as
\begin{equation}\label{eq:graded-function-expansion}
a = \sum_{M,I}a_{MI}\theta^M \xi^I,
\end{equation}
where the coefficients $a_{MI}$ are smooth functions that only depend on the $x^\mu$ coordinates, $M = (\mu_1,\dots,\mu_m)$ and $I = (i_1,\dots,i_\ell)$ are ordered multiindices (i.e., $1\le\mu_1 <\mu_2<\dots <\mu_m\le n$ and $1\le i_1<i_2<\dots<i_\ell \le r$), and
\begin{equation}
\theta^M = \theta^{\mu_1} \dots \theta^{\mu_m},\quad \xi^I = \xi^{i_1}\dots\xi^{i_\ell}.
\end{equation}

\subsection{Graded Vector Fields}\label{sec:graded-vector-fields}
A vector field on the graded manifold can be defined as a first order differential operator, that is, an operator $X$ satisfying the graded Leibnitz rule
\begin{equation}
X(ab) = X(a)b + (-1)^{|a| |X|} aX(b),
\end{equation}
where $|a|$ denotes the total degree of $a$. Any such vector field can be expanded in local coordinates according to
\begin{equation}
X = X^\mu_{(x)} \pdiff{}{x^\mu} + X^\nu_{(\theta)} \pdiff{}{\theta^\nu} + X^i_{(\xi)}\pdiff{}{\xi^i},
\end{equation}
where $\pdiff{}{\theta^\nu}$ has degree $-1$ in $\theta$, $\pdiff{}{\xi^i}$ has degree $-|\xi^i|$ in $\xi$, and the coefficients $X^\mu_{(x)}, X^\nu_{(\theta)}, X^i_{(\xi)}$ are graded functions that may be expanded according to equation~\eqref{eq:graded-function-expansion}. The \emph{graded commutator} of two graded vector fields is defined as
\begin{equation}
    [X,Y] \coloneqq X\circ Y - (-1)^{|X||Y|} Y\circ X,
\end{equation}
and is again a graded vector field of degree $|X|+|Y|$. Note that an odd vector field does not necessarily fulfill $[X,X] = 0$. If it does, $|X|$ is called \emph{homological}.\par
As an example, note that the exterior derivative on $M$ is a homological vector field of degree $+1$ on $T[1]M$: Let us choose local coordinates $(x^\mu,\theta^\nu = \diff x^\nu)$. We may then express $\diff$ locally as
\begin{equation}
\diff = \theta^\mu \partial_\mu.
\end{equation}
We would like to extend this vector field to some vector field $D$ on $\mathcal E$ of degree 1 in $\theta$. To do so, we make the generic ansatz
\begin{equation}
Dx^\mu = \theta^\mu,\quad D \theta^\nu = 0, \quad D\xi^i = \tensor\phi{_\mu_j^i}\theta^\mu \xi^j,
\end{equation}
or equivalently,
\begin{equation}
D = \theta^\mu\pdiff{}{x^\mu} + \theta^\mu\tensor{\phi}{_\mu_j^i}\xi^j \pdiff{}{\xi^i}.
\end{equation}
Note the sign convention: The coefficients $\tensor{\phi}{_\mu_m^\ell}$ act on coordinates with upper indices. If the graded manifold is also equipped with dual coordinates $\psi_1,\dots,\psi_n$, we have $D\psi_i = -\theta^\mu\tensor{\phi}{_\mu_i^k} \psi_k$ such that $D(\xi^i\psi_i) = 0$.\par
We would like to see how this vector field transforms under a change of the $\xi$ coordinates $\xi^i\to (\xi')^i = A^i_j \xi^j$, where $A$ is an invertible matrix that only depends on the $x$ coordinates. We introduce the shorthand notation $\xi = (\xi^1,\dots,\xi^r)$, such that $\xi' = A\xi$. Demanding linearity and the Leibnitz rule, we find
\begin{equation}
D' \xi' = (D \xi )' = A (D \xi) = A D (A^{-1}\xi') = A D (A^{-1}) \xi' + D \xi',
\end{equation}
or $D' = AD(A^{-1}) + D$. It follows that the coefficients $\tensor\phi{_\mu_k^i}$ transform as
\begin{equation}\label{eq:d-extension-transformation}
\tensor\phi{_\mu_k^i} \to A^i_\ell \tensor\phi{_\mu_m^\ell} (A^{-1})^m_k + A^i_\ell \pdiff{}{x^\mu}(A^{-1})^\ell_k,
\end{equation}
that is, the extension $D$ of $\diff$ corresponds precisely to a linear connection on $E$. We shall call such vector fields \emph{$\diff$-extensions}. Note that in contrast to $\diff$, $D$ is generally not homological. Instead, we have
\begin{equation}
D^2 = \frac12[D,D] = (\partial_\mu \tensor\phi{_\nu_k^\ell} + \tensor\phi{_\mu_k^m} \tensor\phi{_\nu_m^\ell})\theta^\mu\theta^\nu \xi^k \pdiff{}{\xi^\ell} \eqqcolon \frac12 \tensor{(F_D)}{_\mu_\nu_k^\ell} \theta^\mu\theta^\nu\xi^k\pdiff{}{\xi^\ell},
\end{equation}
where $F_D$ measures how `non-homological' $D$ is. This function may be identified with the curvature of the connection corresponding to $D$. For the components, we assume without loss of generality that $\tensor{(F_D)}{_\mu_\nu_k^\ell} = - \tensor{(F_D)}{_\nu_\mu_k^\ell}$.
\subsection{Graded Poisson Structures}
Let us briefly discuss some of the defining properties of graded Poisson structures. More details may be found in~\cite{roytenberg2002structure,cattaneo2006graded,basile2025supersymmetric} and references therein.
We consider the graded manifold $\mathcal E$, and write $ (X^A) =(x^\mu, \theta^\nu,\xi^k)$ for the local coordinates. We also fix a not necessarily positive integer $\ell \in\mathbb Z$ (indeed, in the gravitational theories we treat, $\ell$ will be $-2)$.
A \emph{graded Poisson bracket of degree $\ell$} is a bracket
\begin{equation}
    \{-,-\}:\mathcal O_\mathcal E^p \times \mathcal O_\mathcal E^q \to \mathcal O_\mathcal E^{p+q+\ell}, 
\end{equation}
satisfying graded skew-symmetry $\{a,b\} = -(-1)^{(|a|+\ell)(|b| + \ell)}\{b,a\}$, the graded Leibniz rule
\begin{equation}
    \{ab,c\} = a\{b,c\} + (-1)^{|b||a|}b\{a,c\},
\end{equation}
and the graded Jacobi identity
\begin{equation}
    \{a,\{b,c\}\} = \{\{a,b\},c\} + (-1)^{(|a|+\ell)(|b|+\ell)}\{b,\{a,c\}\}.
\end{equation}
Such a bracket may be encoded in a bivector
\begin{equation}
    \pi = \frac12\pi^{AB}\pdiff{}{X^A}\wedge \pdiff{}{X^B},
\end{equation}
and the graded Jacobi identity is equivalent to $[\pi,\pi] = 0$, where $[-,-]$ is the Schouten bracket of vector fields~\cite{cattaneo2006graded}. On the graded coordinates, the Poisson bracket reads
\begin{equation}
    \{X^A,X^B\} = \pi^{AB},
\end{equation}
and may be extended to more general functions by means of the Leibniz rule.

A \emph{Poisson vector field} is a vector field which is compatible with the bracket, that is,
\begin{equation}
    D\{a,b\} = \{Da,b\} + (-1)^{|D|(k+|a|)}\{a,Db\},
\end{equation}
or equivalently, $[D,\pi] = 0$ using the Schouten bracket. When we come to the gravitational theories, we will be interested in $\diff$-extensions that are Poisson vector fields, henceforth referred to as~\emph{Poisson $\diff$-extensions}.

\subsection{Berezin Integrals}\label{sec:berezin-integrals}
Let us briefly summarize the mechanism behind \emph{Berezin integration}. This is a notion of computing integrals on graded manifolds that are concentrated in odd degrees and degree 0. There seems to be no consensus as to how integrals over non-zero even degree directions should be formulated, but fortunately, we do not require this in our treatment.

The flat case is well-known in physics for its application in quantum field theory, and a covariant formulation thereof is summarized in the appendix. We now introduce Berezin integrals on graded manifolds that are curved and do not have a preferred set of coordinates.
Given a curved graded manifold $\mathcal M$, the \emph{orientation-twisted Berezinian line bundle} $I(\mathcal M) \to \mathcal M$ can be constructed. The
\emph{density sheaf} is defined as $\mathcal I_\mathcal M = \Gamma(I(\mathcal M))$, and its sections can be integrated invariantly. The defining property of $I(\mathcal M)$ is that sections $\mu\in\mathcal I_\mathcal M$ transform under diffeomorphisms $\Phi$ as follows,
\begin{equation}
	\Phi^*(\mu) = \operatorname{or(\Phi)}\operatorname{ber}(J\Phi) \mu \circ\Phi,
\end{equation}
where $\operatorname{or}(\Phi)$ is the orientation of $\Phi$, $J\Phi$ is the Jacobian of $\Phi$, and $\operatorname{ber}(J\Phi)$ is the Berezinian of $J\Phi$. For details, see appendix~\ref{sec:app-berezin}, where the integral operator is constructed explicitly in local coordinates. The global integral operator may be defined in charts, where it reduces to the local description. By diffeomorphism invariance of the local picture, we can use a partition of unity to glue the local integral operators to a global integral operator
\begin{equation}
    \int_\mathcal M: \mathcal  I_\mathcal M \to \mathbb R.
\end{equation}
A nowhere vanishing section $\mu\in\mathcal I_\mathcal M$ provides an isomorphism
\begin{equation}
    \mathcal O_\mathcal M \to \mathcal I_\mathcal M,\quad a\mapsto a\mu,
\end{equation}
allowing us to integrate functions. Indeed, we find
\begin{equation}
    \Phi^*(a\mu) = \Phi^*(\mu)(a\circ\Phi) = \operatorname{or}(\Phi) \operatorname{ber}(D\Phi) (a\mu)\circ\Phi.
\end{equation}
Since $\operatorname{ber}(T\mathcal M)$ is a line bundle, any choice of local coordinates $(x^i,\xi^j)$ on $\mathcal U \cong U\times\mathbb R^s[k]$ allows us to write $\mu|_U = \rho [\diff\xi\,\diff x]$ for some function $\rho\in\mathcal O_\mathcal M(U)$, and
\begin{equation}
    \int_\mathcal U a \vol = 
    \int_U \iota^*\left(\pdiff{}{\xi^1}\dots\pdiff{}{\xi^s } (a\rho)\right) d^n x,
\end{equation}
where $d^nx =d\lambda (x)$ is the Lebesgue measure, and $\iota: U\to\mathcal U$ is the inclusion.
\paragraph{On $\mathbf{T[1]M}$.} Recall that any diffeomorphism $\Phi$ of $M$ extends to a diffeomorphism $(\Phi,J\Phi)$ of $T[1]M$. Indeed, under $x^\mu \to y^\mu = \Phi^\mu(x)$, we have $\theta^\mu \to \chi^\mu = (J\Phi)^\mu_\nu \theta^\nu$.

There exists a canonical Berezin measure on $T[1]M$ which is invariant under such diffeomorphisms $(\Phi,J\Phi)$. In local coordinates $(x^\mu, \theta^\nu = \diff x^\nu)$, it is explicitly given by
\begin{equation}
[dx^1\,\dots\, dx^n\, d\theta^1\,\dots\, d\theta^n].
\end{equation}
To verify invariance, we use description~\eqref{eq:berezin-pdiff} from appendix~\ref{sec:app-berezin}, and compute
\begin{gather}
\diff x^1\dots \diff x^n \to \diff y^1 \dots \diff y^n = \det(J\Phi)\,\diff x^1\dots \diff x^n,\\
\pdiff{}{\theta^s}\dots\pdiff{}{\theta^1} \to \pdiff{}{\chi^s}\dots\pdiff{}{\chi^1} 
= \det (J\Phi)^{-1}\pdiff{}{\theta^s}\dots\pdiff{}{\theta^1},
\end{gather}
such that the factor $\det(J\Phi)$ cancels out.\par
In local coordinates, functions on $T[1]M$ are integrated according to
\begin{equation}
\int_{T[1]M} \sum_N a_N \theta^N  [dx^1\,\dots\, dx^n\, d\theta^1\,\dots\, d\theta^n] = \int_M \left(\pdiff{}{\theta^n}\dots\pdiff{}{\theta^1} a \right)[dx^1\dots dx^n] = \int_M a_{1\dots n} d^n x,
\end{equation}
where the sum ranges over ordered multiindices $N$.
Thus, locally, Berezin integration on $T[1]M$ entails replacing $\theta^\mu$ with $dx^\mu$, and subsequently computing the Lebesgue integral. This amounts precisely to the invariant integration of differential forms on $M$.

\paragraph{On $\mathbf{T[1]M \oplus E[k]}$.}
Let us now consider the graded manifold $\mathcal E = T[1]M \oplus E[k]$ for $k$ odd. Let $r$ be the rank of $E$, such that $\Gamma(M,E)$ can be generated locally over $C^\infty(M)$ by a local basis $(\xi_1,\dots,\xi_r)$. Dually, $\mathcal E$ comes with coordinates $(x^\mu, \theta^\nu, \xi^i)$. This means that an arbitrary function $f\in\mathcal O_\mathcal E$ may be expanded locally in graded coordinates as
\begin{equation}
a = \sum_{M,I} a_{MI}\theta^M \xi^I
\end{equation} 
By picking a generator $\omega = \frac1{r!}\omega^{i_1\dots i_r} \xi_{i_1} \dots \xi_{i_r} = \rho \xi_1\dots\xi_r$ of the determinant line $\det(E)$, we can define a Berezin measure $\mu$ on $\mathcal E$ by setting
\begin{equation}
\mu = [dx\, d\theta\, \omega] = \rho [dx\,d\theta\, d\xi].
\end{equation}
More explicitly, integrals may be computed as follows
\begin{equation}
\int_\mathcal M a \mu_g = \int_{T[1]M} \left(\pdiff{}{\xi^r}\dots\pdiff{}{\xi^1} a \rho \right)[dx\,d\theta],
\end{equation}
with the integral over $T[1]M$ computed as above.

\section{Graded Einstein-Cartan Gravity}\label{sec:metric-gravity}
We now specialize the setup to the case where a (pseudo) Riemannian metric on $M$ is the dynamical field. To this end, it turns out to be useful to consider the graded manifold $\mathcal M = T[1]M\oplus T[1]M$ with coordinates $(x^\mu, \theta^\nu, \tilde\theta^\lambda)$. As before, we use holonomic degree 1 coordinates $\theta^\mu = \diff x^\mu$, and likewise $\tilde\theta^\nu = \widetilde{\diff x^\nu}$. On this graded manifold, the metric may be intrinsically understood in the context of graded Poisson geometry.

\subsection{\texorpdfstring{The Graded Manifold ${T[1]M\oplus T[1]M}$}{The Graded Manifold T[1]M+T[1]M}}
In contrast to the Minkowski case, the coordinates $(x^\mu,\theta^\nu,\tilde\theta^\lambda)$ can only be chosen locally. Indeed, the algebra of functions on $\mathcal M $ is locally isomorphic to
\begin{equation}
C^\infty(U)[\theta^\mu,\tilde\theta^\nu] = C^\infty(U)\otimes \left(\extp (\mathbb R^n)^* \right)^{\otimes 2},
\end{equation}
where $U\subset M$ is a contractible patch, and the graded coordinates $\theta^\nu = \diff x^\mu$, $\tilde\theta^\lambda = \widetilde{\diff x^\lambda}$ coming from the fiber directions of the first resp.~second copy of $T[1]M$ satisfy
\begin{equation}
\theta^\nu\theta^\lambda = -\theta^\lambda\theta^\nu,\quad \tilde\theta^\nu\tilde\theta^\lambda = -\tilde\theta^\lambda\tilde\theta^\nu,\quad \tilde\theta^\nu\theta^\lambda = -\theta^\lambda\tilde\theta^\nu.
\end{equation}
Graded functions on $\mathcal M$ carry two independent degrees, coming from the degree in the $\theta$ and the $\tilde\theta$ coordinates. We write $\mathcal O_\mathcal M$ for the algebra of functions, and $\mathcal O^{p,q}_\mathcal M$ for functions that are homogeneous of degree $p$ in $\theta$ and homogeneous of degree $q$ in $\tilde\theta$. 
\subsection{Geometric Structures}\label{sec:geometric-structures}
We now discuss the graded geometric structures on $\mathcal M$. We will characterize these structures intrinsically in the language of graded geometry.
\paragraph{The Transposition Operation.}
The algebra of functions $\mathcal O _\mathcal M$ comes with an invariant involution which exchanges the roles of $\theta$ and $\tilde\theta$. More explicitly, it is defined on the generators by $\widetilde{\theta^\mu} = \tilde\theta^\mu$ and $\widetilde{\tilde\theta^\mu} = \theta^\mu$. This extends to arbitrary graded functions as follows: Let $a = \sum_{MN}a_{MN}\theta^M \tilde\theta^N$. Then
\begin{equation}
\tilde a = \sum_{MN}a_{MN}\tilde\theta^M \theta^N = \sum_{MN}{(-1)^{|M||N|}}a_{MN} \theta^N \tilde\theta^M =
\sum_{MN}{(-1)^{|M||N|}}a_{NM} \theta^M \tilde\theta^N,
\end{equation}
where we renamed the summation indices in the final step. Thus, for the components, it holds $ \tilde a_{MN} = (-1)^{|M||N|} a_{NM}$.
\paragraph{The Poisson bracket.}
For graded Einstein-Cartan gravity, we require a Poisson structure of degree $-2$ in $\tilde\theta$ with symplectic leaves given precisely by the submanifolds along fixed $x$ and $\theta$. The  structure may be encoded in a bivector
\begin{equation}
	\pi = \frac12 g^{\mu\nu}\pdiff{}{\tilde\theta^\mu}\wedge\pdiff{}{\tilde\theta^\nu},
\end{equation}
where the components $g^{\mu\nu}$ are invertible. In other words, the Poisson structure corresponds uniquely to an inverse metric on $M$. 
The bivector clearly fulfills $[\pi,\pi] = 0$ as the coefficients $g^{\mu\nu}$ do not depend on the $\tilde\theta$ coordinates. The non-vanishing brackets of the coordinate functions read
\begin{equation}
\{\tilde\theta^\mu,\tilde\theta^\nu\} = g^{\mu\nu}.
\end{equation}
Brackets of higher degree objects can be calculated inductively using the graded Leibnitz rule. As the components $g^{\mu\nu}$ are invertible, we may associate a \emph{structure function}
\begin{equation}
    g_{\mu\nu}\theta^\mu\tilde\theta^\nu\in\mathcal O^{1|1}_\mathcal M
\end{equation}
with the Poisson bivector. We also have a natural associated second order differential operator
\begin{equation}\label{eq:square-operator}
\trg = g^{\mu\nu}\pdiff{}{\theta^\mu}\pdiff{}{\tilde\theta^\nu},
\end{equation}
which pairs $\theta$ with $\tilde\theta$ and can be viewed as a metric trace.
\par
The degree 2 functions of $\tilde\theta^\mu$ generate the Lorentz algebra associated with $g$, and their action on the $\tilde\theta^\mu$ coordinates via $\rho(\lambda) f = \{\frac12\lambda_{\mu\nu}\tilde\theta^\mu\tilde\theta^\nu, f\}$ reflects the action of the Lorentz algebra on $\mathcal O_\mathcal M$.
We can also encode the Clifford algebra action on the $\tilde\theta^\mu$ coordinates: Clifford multiplication is given by 
\begin{equation}
	c(\tilde\theta^\mu)(f) = \tilde\theta^\mu f + \{\tilde\theta^\mu,f\},
\end{equation}
such that 
\begin{equation}
	c(\tilde\theta^\mu)\circ c(\tilde\theta^\nu) + c(\tilde\theta^\nu)\circ c(\tilde\theta^\mu)  = 2g^{\mu\nu}.
\end{equation}
We remark at this point that the Clifford action breaks the $\mathbb Z$ grading down to a mere $\mathbb Z_2$ grading as Clifford multiplication only preserves the latter.

Analogously, we can consider the bivector $\pi' = \frac12 g^{\mu\nu}\pdiff{}{\theta^\mu}\wedge\pdiff{}{\theta^\nu}$ with associated Poisson bracket $\{\theta^\mu,\theta^\nu\}' = g^{\mu\nu}$. This bracket has degree $-2$ in $\theta$, and can be used to define the Clifford and Lorentz action on the $\theta^\mu$ coordinates.\par
Note that this structure intertwines the action of diffeomorphisms with that of Lorentz transformations. The different roles of $\theta$ and $\tilde\theta$ are made explicit in the vielbein formalism, which we derive in section~\ref{sec:vielbein-gravity}.
\paragraph{Poisson $\mathbf\diff$-extensions.}\label{par:poisson-d-ext}
In a given coordinate chart, we can define the local vector fields
\begin{equation}
\diff = \theta^\mu\pdiff{}{x^\mu},\quad \widetilde\diff =\tilde\theta^\mu \pdiff{}{x^\mu}
\end{equation}
of degree $(1,0)$ resp. $(0,1)$. However, these vector fields do not transform correctly under a change of odd coordinates, hence, they do not define global vector fields on $\mathcal M$. Instead, we shall require suitable `covariantizations' of these differential operators. To that end, we make the general ansatz
\begin{equation}
D = \theta^\mu \pdiff{}{x^\mu} + \theta^\mu\tensor{\phi}{_\mu_\nu^\lambda}\tilde\theta^\nu \pdiff{}{\tilde\theta^\lambda}
\end{equation}
as in section~\ref{sec:graded-vector-fields}, and additionally impose the following compatibility condition with the Poisson bracket
\begin{equation}
D\{f,g\} = \{D a, b\} + (-1)^{|a|} \{a, Db\}.
\end{equation}
In other words, we consider extensions of $D$ that are Poisson derivations. By letting $a=\tilde\theta^\kappa$ and $b=\tilde\theta^\lambda$, we can extract the following relation
\begin{equation}\label{eq:poisson-metricity}
\theta^\mu\partial_\mu g^{\kappa\lambda} = 
\{\tensor{\phi}{_\rho_\sigma^\kappa} \theta^\rho\tilde\theta^\sigma , \tilde\theta^\lambda\} -
\{\tilde\theta^\kappa , \tensor{\phi}{_\rho_\sigma^\lambda} \theta^\rho\tilde\theta^\sigma\} = \theta^\mu \tensor{\phi}{_\mu_\rho^\kappa} g^{\rho\lambda} + \theta^\mu \tensor{\phi}{_\mu_\rho^\lambda} g^{\kappa\rho},
\end{equation}
which is the case if and only if $\varphi$ corresponds to a metric connection ($\nabla g = 0$). In the following, we shall always assume that $D$ is compatible with the Poisson structure. We will call these derivations \emph{Poisson $\diff$-extensions}.\par
Let us evaluate the implications for the curvature. We have
\begin{equation}
    0 = 
    D^2\{\tilde\theta^\kappa,
    \tilde\theta^\lambda\} =
    \{D^2\tilde\theta^\kappa,
    \tilde\theta^\lambda\} + \{\tilde\theta^\kappa,
    D^2\tilde\theta^\lambda\} = \tensor{(F_D)}{_\mu_\nu_\rho^\kappa} g^{\rho\lambda} +
    \tensor{(F_D)}{_\mu_\nu_\rho^\lambda} g^{\kappa\rho},
\end{equation}
that is, $\tensor{(R_D)}{_\mu_\nu_\kappa_\lambda} + \tensor{(R_D)}{_\mu_\nu_\lambda_\kappa} = 0$, where $\tensor{(R_D)}{_\mu_\nu_\kappa_\lambda} = \tensor{(F_D)}{_\mu_\nu_\kappa^\rho}g_{\rho\lambda}$. Let us define
\begin{equation}
    R_D = \frac14 \tensor{(F_D)}{_\mu_\nu_\kappa^\rho} g_{\rho\lambda} \theta^\mu\theta^\nu\tilde\theta^\kappa\tilde\theta^\lambda.
\end{equation}
It follows that 
\begin{equation}
    D^2 a = \{R_D,a\}
\end{equation}
for all graded functions $a$. In plain English, $D^2$ is a Hamiltonian vector field, and $R_D$ is the corresponding Hamiltonian. Since the coefficients $\varphi$ act without the factor of $-1$ on $\tilde\theta^\mu = \widetilde{\diff x}^\mu$, the function $R_D$ corresponds to \emph{minus} the Riemann tensor in the conventions of~\cite{misner1974gravitation}. We will revisit this when we come to the graded gravity theories.\par
Now, let $D'$ be another $\diff$-extension. Denote \begin{equation}
   \delta\phi = D'-D =  \theta^\mu\left(\tensor{(\phi')}{_\mu_\nu^\lambda} -\tensor{\phi}{_\mu_\nu^\lambda}\right) \tilde\theta^\nu \pdiff{}{\tilde\theta^\lambda}
\end{equation}
Then $D'$ is Poisson if and only if
\begin{equation}
    0 =\delta\phi\left(\{\tilde\theta^\nu,\tilde\theta^\lambda\}\right) = \theta^\mu \tensor{(\delta\phi)}{_\mu_\rho^\nu} g^{\rho\lambda} + \theta^\mu \tensor{(\delta\phi)}{_\mu_\rho^\lambda} g^{\kappa\rho},
\end{equation}
or equivalently, $\tensor{(\delta\varphi)}{_\mu_\nu_\lambda} + \tensor{(\delta\varphi)}{_\mu_\lambda_\nu}$, where $(\delta\varphi)_{\mu\nu\lambda} = \tensor{(\delta\phi)}{_\mu_\nu^\kappa} g_{\kappa\lambda}$. In this case, we may write
\begin{equation}
    \delta\varphi (a) = \left\{\frac12\theta^\mu\tensor{(\delta\varphi)}{_\mu_\nu_\lambda}\tilde\theta^\nu\tilde\theta^\nu, a\right\}.
\end{equation}
We conclude that the difference of any two Poisson $\diff$-extensions is a Hamiltonian vector field. If $D$ is Poisson, we will write $\varphi$ instead of $\phi$ for the coefficients.\par
Finally, we note that $Dg = \frac12T_{\mu\nu\kappa}\theta^\mu\theta^\nu\tilde\theta^\kappa$, where $T_{\mu\nu\kappa} = \tensor{T}{_\mu_\nu^\lambda} g_{\kappa\lambda}$ are the components of the torsion of $\nabla$.
\paragraph{The Berezin Measure and Stokes' Theorem.}
As discussed in section~\ref{sec:berezin-integrals}, we require a Berezin measure to integrate functions on graded manifolds. To that end, we use $g$ and define
\begin{equation}
	\mu_g = |\det(g_{\mu\nu})|^{-\frac12} [dx\, d\theta\, d\tilde\theta],
\end{equation}
where, as before, we chose the odd coordinates are chosen such that $\theta^\mu = \diff x^\mu$ and $\tilde\theta^\nu= \widetilde{\diff x^\nu}$. 
We have the following graded incarnation of Stokes' theorem:
\begin{theorem}\label{thm:stokes-M}
Let $\mu_g$ be the Berezin measure constructed above, let $f\in\mathcal O_\mathcal M$, and denote by
\begin{equation}
\mathcal M^\partial = T[1]\partial M \oplus T[1]M|_{\partial M}
\end{equation}
with induced Berezin measure
\begin{equation}
\mu^\partial_g = |\det(g)|^{-\frac12}[dy\, d\psi\, d\tilde\theta],
\end{equation}
where $(y^\mu,\psi^\nu = \diff y^\nu)$ are coordinates on $T[1]\partial M$, and $\tilde\theta^\lambda$ are the restricted degree 1 coordinates on $T[1]M|_{\partial M}$. It holds
\begin{equation}\label{eq:stokes-M}
\int_\mathcal M (Da)\mu_g = \int_{\mathcal M^\partial} a\mu_g^\partial.
\end{equation}
\end{theorem}
\begin{proof}
We would like to apply theorem~\eqref{thm:stokes-graded-general}. To do so, we need to show that $L_D (\mu_g) = 0$. Recall that on functions, $L_D(f) = D(f)$. The coordinate functions thus satisfy
\begin{equation}
L_D x^\mu = \theta^\mu, \quad L_D\theta^\nu = 0,\quad L_D \tilde\theta^\lambda =  \tensor\varphi{_\mu_\nu^\lambda} \theta^\mu\tilde\theta^\nu.
\end{equation}
The partial measure $[dx\,d\theta]$ is invariant as $L_D[dx \,d\theta] = \pdiff{\theta^\mu}{x^\mu}[dx\,d\theta] = 0$, and so is
\begin{equation}
|\det(g_{\mu\nu})|^{\frac12}\tilde\theta^1\dots\tilde\theta^n |\det(g_{\mu\nu})|^{-\frac12}[d\tilde\theta^1\dots d\tilde\theta^n]
\end{equation}
as the $\det(g_{\mu\nu})$ factors cancel, and $\tilde\theta^\mu$ transforms precisely dually to $[d\tilde\theta^\mu]$.
Hence, instead of computing the Lie derivative of the measure, we may calculate the simpler expression $D(|\det(g_{\mu\nu})|^{\frac12}\tilde\theta^1\dots\tilde\theta^n)$. We first evaluate
\begin{equation}
L_D(\tilde\theta^1\dots\tilde\theta^n) = \theta^\mu\left(\sum_\lambda \tensor{\varphi}{_\mu_\lambda^\lambda}\right)\tilde\theta^1\dots\tilde\theta^n.
\end{equation}
By multiplying equation~\eqref{eq:poisson-metricity} with $g_{\kappa\lambda}$, we have $D\tr(g_{\kappa\lambda}) = g^{\kappa\lambda}D g_{\kappa\lambda} = -2\tr(\varphi)$.
On the other hand, we have
\begin{equation}
L_D\left(|\det(g_{\mu\nu})|^{-\frac12}\right) = \frac12 \theta^\mu\partial_\mu\left(\left(\tr(g^{\mu\nu}) \right)|\det(g_{\mu\nu})|^{-\frac12}\right),
\end{equation}
proving that $L_D(|\det(g_{\mu\nu})|^{\frac12}\tilde\theta^1\dots\tilde\theta^n) = 0$. Putting everything together, it follows that the Lie derivative of the total Berezin measure vanishes $L_D(|\det(g_{\mu\nu})|^{-\frac12}[dx\,d\theta\,d\tilde\theta])=0$, such that equation~\eqref{eq:stokes-M} follows directly from theorem~\eqref{thm:stokes-graded-general}.
\end{proof}
Now, let $g'$ be a nearby Poisson structure such that $g'= g+\delta g$, which we may parametrize by some symmetric $\tensor{\xi}{_\mu^\nu}$ such that to first order,
\begin{equation}
g'_{\mu\nu} = (\tensor\delta{_\mu^\kappa}+\tensor\xi{_\mu^\kappa})g_{\kappa\lambda} (\tensor\delta{_\nu^\lambda}+\tensor\xi{_\nu^\lambda})= g_{\mu\nu} + g_{\mu\lambda} \tensor\xi{_\nu^\lambda} + \tensor\xi{_\mu^\kappa}g_{\kappa\nu} + \mathcal O(\xi^2),
\end{equation}
from which we can read off that $\delta g_{\mu\nu} = 2\tensor{\xi}{_\mu^\kappa} g_{\kappa\nu}$. Again, to first order, it follows that
\begin{equation}\label{eq:volume-variation}
\mu_{g'} = |\det((\tensor{\delta}{_\mu^\kappa} + 2\tensor\xi{_\mu^\kappa})g_{\kappa\nu})|^{-\frac12}[dx\,d\theta\,d\tilde\theta] = \left(1 -\frac12 \trg(\delta g)\right)\mu_g + \mathcal O (\xi^2),
\end{equation}
where we used multiplicativity of the determinant as well as
\begin{equation}
\det(1+\xi^T) = \det(1+\xi) = \tr(\xi) + \mathcal O(\xi^2)
\end{equation}
together with equation~\eqref{eq:square-operator}. Here, $\tr(\xi)=\sum_\mu \tensor\xi{_\mu^\mu}$ may be interpreted as the trace of the coefficients, or, equivalently, as a divergence with respect to the measure $[d\theta^1\dots d\theta^n]$ induced by the choice of local coordinates.
\paragraph{The $\star$ Operator.}
In~\cite{chatzistavrakidis2020graded}, the authors introduced a generalized $\star$ operator for functions on $\mathcal M$, where the underlying ordinary manifold $M$ was Minkowski space. We now extend this definition to the curved case and to functions of arbitrary bidegree by declaring $\star:\mathcal O_\mathcal M^{p,q} \to \mathcal O_\mathcal M^{n-p,n-q}$  as
\begin{equation}
\star a = \begin{cases}
    \frac{1}{(n-p-q)!}g^{n-p-q}\tilde a, & p+q\le n,\\
    \frac{1}{(p+q-n)!}\trg^{p+q-n}\tilde a, &p+q>n.
\end{cases} 
\end{equation}
As shown in appendix~\ref{app:properties-star}, the operator $\star^2$ is has non-zero eigenvalues. It follows that $\star:\mathcal O^{p,q}_\mathcal M \to \mathcal O^{n-p,n-q}_\mathcal M$ is an isomorphism.

\subsection{Graded Einstein-Cartan Gravity}
Employing the graded geometric structures introduced above, the gravity action functional can be concisely expressed on $\mathcal M = T[1]M\oplus T[1]M$ as
\begin{equation}\label{eq:graded-gravity}
S(g,D) = \frac{1}{(n-2)!} \int_\mathcal M  R_D\, g^{n-2}\,\mu_g  + S_m(g,D,\dots),
\end{equation}
where $\mu_g = |\det(g)|^{-\frac12}[dx\,d\theta\,d\tilde\theta]$ is the Berezin measure constructed from $g$, $D$ is a Poisson $\diff$-extension. The prefactor is chosen such that the variational principle gives the same weight to $g$ and $D$ as in the standard Einstein-Cartan action with unit prefactor.\footnote{That is, $S=\int_M s^\varphi\sqrt{-g_{\mu\nu}} \,d^nx + S_m$ with $s^\varphi = \tensor{R}{_\mu_\nu_\kappa^\nu}g^{\mu\kappa}$ the scalar curvature. Recall that in our convention, $D$ acts on $\theta,\tilde \theta$, which have upper indices. Therefore, $\tensor{R}{_\mu_\nu_\kappa^\lambda}$ differs in sign from that of, e.g.,~\cite{misner1974gravitation}}
The generic matter functional $S_m(g,D,\dots)$ may depend on the Poisson structure $g$, as well as the Poisson $\diff$-extension $D$. Since the functional does not depend on auxiliary fields, it is diffeomorphism invariant by Stokes' theorem.\par
Note that for $d\ge 4$, the pure gravity functional may be rewritten using the $\star$ operator, taking the compact form
\begin{equation}
    S_\textrm{EC}(g,D) = \frac1{(n-2)!}\int_\mathcal M  R_D\, g^{n-2}\,\mu_g = \frac1{(n-2)(n-3)} \int_\mathcal M \left( R_D\star
g^2\right)\mu_g.
\end{equation}
Since the metric is a dynamical field of the gravitational theory, we shall resort to the former formulation of the action functional which makes the dependence of $g$ explicit.
To derive the field equations, we need to relate the variation of the pure gravity functional $S_\textrm{EC}$ to that of the matter functional using equation~\eqref{eq:graded-gravity}.
To that end, we define
\begin{equation}
    \delta S_\textrm{EC} = \int_{\mathcal M} \left(\mathcal G\,\delta g  + \mathcal H\,\delta\varphi\right)\mu_g,\quad
    \delta S_m = \int_{\mathcal M} \left( \mathcal T\,\delta g + \mathcal S\, \delta\varphi \right)\mu_g,
\end{equation}
where $\mathcal T$ and $\mathcal S$ denote stress-energy and spin of the matter, respectively. Since $g = -\widetilde{g}$, we fix the freedom in choosing $\mathcal T$ and $\mathcal G$ such that
\begin{equation}
    \widetilde{\mathcal G} = (-1)^{n+1}\mathcal G,\quad \widetilde{\mathcal T} = (-1)^{n+1}\mathcal T.
\end{equation}
\par
We now compute the variation of the pure gravity functional with respect to $g$ as follows: Any other $g'$ may be written as
\begin{equation}
g'_{\mu\nu} = (A^* g)_{\mu\nu} = A_\mu^\kappa g_{\kappa\lambda} A_\nu^\lambda
\end{equation}
for some invertible matrix $A_\mu^\kappa$ (indeed, any non-degenerate bilinear form is uniquely determined by the datum of an orthonormal basis). Since $\mu_{A^*g} = |{\det(A)}|^{-1}\mu_g$, we have
\begin{equation}
\begin{aligned}
 S_\textrm{EC}(A^*g, D) &= \frac{1}{(n-2)!}\int_{\mathcal M }(A^*g)^{n-2}\cdot \frac{1}{4}\tensor{F}{_\kappa_\lambda_\rho^\tau}(A^*g)_{\tau\sigma} \theta^\kappa\theta^\lambda\tilde\theta^\rho\tilde\theta^\sigma\cdot \frac1{|\det(A)|}\mu_g\\ 
&= \frac1{(n-2)!}\int_\mathcal M \left(g_{\mu\nu} (A\theta)^\mu (\widetilde{A\theta})^\nu \right)^{n-2} \cdot \frac14 R_{\kappa\lambda\rho\sigma}\theta^\kappa \theta^\lambda (\widetilde{A\theta})^\rho (\widetilde{A\theta})^\sigma \cdot \frac{|\det(g)|^{-\frac12}}{|\det(A)|} [dx\,d\theta\,d\tilde\theta]\\
&= \frac1{4(n-2)!}\int_\mathcal M \left(g_{\mu\nu} (A\theta)^\mu (\widetilde{A\theta})^\nu \right)^{n-2}R_{\kappa\lambda\rho\sigma}\theta^\kappa \theta^\lambda (\widetilde{A\theta})^\rho (\widetilde{A\theta})^\sigma |\det(g)|^{-\frac12} [dx\,d\theta\,d(\widetilde{A\theta})]\label{eq:action-finite-variation}\\
&= \frac1{4(n-2)!}\int_\mathcal M \left(g_{\mu\nu} (A\theta)^\mu \tilde\theta^\nu \right)^{n-2}R_{\kappa\lambda\rho\sigma}\theta^\kappa \theta^\lambda \tilde\theta^\rho \tilde\theta^\sigma |\det(g)|^{-\frac12} [dx\,d\theta\,d\tilde\theta],
\end{aligned}
\end{equation}
where we performed the coordinate transformation $\tilde\theta \to A^{-1}\tilde\theta$ in the final step. Note that in the last formula, the variation of $R$ and $\mu_g$, together with half of the variation of $g_{\mu\nu}\theta^\mu\tilde\theta^\nu$ with respect to $g$, have been absorbed into the coordinate transformation. Thus, upon taking the adapted variation of $g$, $\mu_g$ may be regarded as a background field, while $R$ may be regarded as independent of $g$.
We now let $A(t)$ be a one-parameter family of invertible matrices with $A_\mu^\kappa(0) = \delta_\mu^\kappa$, and evaluate
\begin{equation}
\ddt S_\textrm{EC}(A(t)^*g, D) = \ddt\frac{1}{(n-2)!} \int_\mathcal M (A(t)^*g)^{n-2}R_{A(t)^* g} \mu_{A^*(t)g}.
\end{equation}
Without loss of generality, we may assume that $A_\mu^\kappa(t) = \delta_\mu^\kappa + t\xi_\mu^\kappa + \mathcal O(t^2)$, where $\xi_\mu^\kappa$ is $g$-symmetric (indeed, the antisymmetric contribution of $\xi$ in $(\delta + t\xi)^*g$ vanishes to first order in $t$).
It follows that
\begin{equation}
\ddt\frac{1}{(n-2)!} \int_\mathcal M (A(t)^*g)^{n-2}R_{A(t)^* g} \mu_{A^*(t)g} = \frac1{(n-3)!}\int_\mathcal M  \delta g\, g^{n-3} \,R_D\, \mu_g,
\end{equation}
where $\delta g \coloneqq \frac12(\xi^\kappa_\mu g_{\kappa\nu} + \xi^\kappa_\nu g_{\mu\kappa})\theta^\mu \tilde\theta^\nu = \xi^\kappa_\mu g_{\kappa\nu}\theta^\mu\tilde\theta^\nu$. The resulting equations of motion are given by
\begin{equation}
\frac{1}{2(n-3)!}g^{n-3}(R+\tilde R)  =  \mathcal T,
\end{equation}
where symmetrization in $\theta,\tilde\theta$ is necessary since $\delta g$ is graded symmetric under $\theta\leftrightarrow \tilde\theta$. Relating this to the classical Lorentzian geometry setup, this is equation tells us that the Einstein tensor must equate to the stress-energy tensor.\par
As established above, after the coordinate transformation, $D$ may be regarded as completely independent of $g$. Performing the variation of $D$, i.e., $D\to D' = D+\{\delta\varphi,-\}$, we find $R_D\to R_D+D(\delta\varphi)$, such that
\begin{equation}
\ddt S_\textrm{EC}(g,D + t\{\delta\varphi,-\}) = \frac{1}{(n-2)!}\int_\mathcal M  g^{n-2}\,D(\delta\varphi)\,\mu_g = \frac{1}{(n-3)!}\int_\mathcal M  \, Dg\, g^{n-3}\, \delta\varphi\,\mu_g,
\end{equation}
provided that
\begin{equation}
\frac1{(n-2)!} \,g^{n-2} \,\delta\varphi= 0\quad\textrm{ on } \mathcal M^\partial.
\end{equation}
Here, $Dg = T$ corresponds to the torsion of $D$ in the classical differential geometric picture, as discussed in section~\eqref{sec:geometric-structures}. The resulting equation of motion for $D$ reads
\begin{equation}
\frac1{(n-3)!} g^{n-3} Dg = \mathcal S.
\end{equation}
Relating this back to classical Lorentzian geometry, we thus find that $D$ must correspond to the Levi-Civita connection.

\section{Graded Palatini-Cartan Gravity}\label{sec:vielbein-gravity}
We now reformulate the gravitational theory discussed above to recover an Einstein-Cartan type theory of gravity. As a starting point, we take the graded gravity from the preceding chapter, and use the dynamical field $g$ to construct an isomorphism of graded manifolds $\Phi:T[1]M\oplus T^*[1]M \xrightarrow\sim T[1]M \oplus T[1]M$. We first translate the geometric structures from $\mathcal M$ to their counterparts on $\mathcal N$, and then show how to formulate the variational principle concisely on $\mathcal N$. Finally, we construct the appropriate BV extension of the theory.
\subsection{\texorpdfstring{The Graded Manifold ${T[1]M\oplus T^*[1]M}$}{The Graded Manifold T[1]M+T*[1]M}}
In the following, we are working on the graded manifold 
\begin{equation}
    \mathcal N = T[1]M \oplus T^*[1] M,
\end{equation}
equipped with local coordinates $(x^\mu, \theta^\nu = \diff x^\nu, \chi_\kappa = \partial_\kappa)$, where the latter is understood fiber coordinates on $T^*[1]M$. Note that $\mathcal N$ comes with a canonical function \begin{equation}
    \vartheta = \theta^\mu\chi_\mu,
\end{equation}
as well as an associated second order differential operator
\begin{equation}
    \Delta_\vartheta = \pdiff{}{\theta^\mu}\pdiff{}{\chi_\mu}
\end{equation}
of degree $-2$.
To relate $\mathcal N$ to $\mathcal M$, we take the field $g$ and construct an isomorphism $\Phi:\mathcal N\xrightarrow\sim \mathcal M$,
given on local coordinates by
\begin{equation}
    \Phi^*(x^\mu) = x^\mu,\quad \Phi^*(\theta^\nu) = \theta^\nu, \quad
    \Phi^*(\tilde\theta^\lambda) = g^{\lambda\kappa} \chi_\kappa.
\end{equation}
\subsection{Geometric Structures}
The graded geometric structures on $\mathcal M$ can now be transferred to $\mathcal N$ via $\Phi$. We denote the geometric structures on $\mathcal N$ by the same symbols as their counterparts on $\mathcal M$; since we work on $\mathcal N$ in this section, this should not lead to confusion.
\paragraph{The Transposition Operation.} The transposition operation may be defined on
$\mathcal N$ via
\begin{equation}
    \widetilde {\theta^\mu} = g^{\mu\nu} \chi_\nu,\quad \widetilde{\chi_\mu} = g_{\mu\nu}\theta^\nu.
\end{equation}
In particular, a function $a\in\mathcal O^{p,p}_\mathcal N$ corresponds to a $g$-symmetric tensor if and only if
\begin{equation}
    \tilde a = (-1)^p a.
\end{equation}
Note that unlike before, the transposition depends on $g_{\mu\nu}$. For the purposes of the gravitational theory that we have in mind, we shall not require the transposition operation.
\paragraph{The $\star$ Operator.}
The $\star$ operator may be defined on $a\in\mathcal O^{p,q}_\mathcal N$ by
\begin{equation}
    \star a = \begin{cases}\frac{1}{(n-p-q)!} \vartheta^{n-p-q} \widetilde a,&n\ge p+q,\\
    \frac{1}{(p+q-n)!} (\Delta_\vartheta)^{p+q-n} \tilde a , &p+q> n.
    \end{cases}
\end{equation}
Note that $\star$ again depends on $g$, this time via the transposition operation. In particular, when $\star$ is applied to some transposition invariant function, it does not depend on $g$. As before, the $\star$ operator on $\mathcal N$ is an isomorphism $\mathcal O^{p,q}_\mathcal N \xrightarrow\sim \mathcal O^{n-p,n-q}_\mathcal N$.
\paragraph{The Poisson Structure.} We now encode the Poisson structure on the $\chi_\mu$ coordinates by
\begin{equation}
    \{\chi_\mu,\chi_\nu\} = g_{\mu\nu}.
\end{equation} 
As before, infinitesimal Lorentz transformations are encoded in this bracket by
\begin{equation}\label{eq:lorentz-vielbein}
    \rho(\lambda)(v) = \left\{\frac12\lambda^{\mu\nu}\chi_\mu\chi_\nu, v^\kappa\chi_\kappa\right\} = \lambda^{\mu\nu} g_{\nu\kappa} v^\kappa \chi_\mu = \lambda\cdot v,
\end{equation}
and extended to $\mathcal O_\mathcal N$ by the graded Leibniz rule. The bracket also gives rise to a Clifford algebra representation by setting
\begin{equation}
    c(\chi_\mu) = \chi_\mu + \{\chi_\mu,-\}.
\end{equation}
A straight-forward computation shows that
\begin{equation}
    [c(\chi_\mu), c(\chi_\nu)] = 2g_{\mu\nu}\cdot 1.
\end{equation}
\paragraph{Poisson $\diff$-extensions.} As in the metric case, we are interested in Poisson $\diff$-extensions, i.e.,
\begin{equation}
    D\{a,b\} = \{Da,b\} + (-1)^{|a|}\{a,Db\}.
\end{equation}
As in eq.~\eqref{eq:poisson-metricity}, we find that this is the case if 
\begin{equation}
    D = \theta^\mu \partial_\mu - \theta^\mu\tensor{\varphi}{_\mu_\nu^\lambda} \chi_\lambda\pdiff{}{\chi_\nu},
\end{equation}
where $\tensor{\varphi}{_\mu_\nu^\lambda}$ are the same coefficients as in paragraph~\ref{par:poisson-d-ext}. We also find
\begin{equation}
    D \vartheta = \frac12\tensor{T}{_\mu_\nu^\lambda} \theta^\mu\theta^\nu \chi_\lambda =: T,
\end{equation}
where $T$ corresponds to the torsion of $\varphi$. Finally, we note that
\begin{equation}
    D^2 a = \{R_D,a\},
\end{equation}
where $R_D = -(\diff\varphi + \frac12\{\varphi,\varphi\})$, $\tensor{R}{_\mu_\nu^\kappa^\lambda}$ are \emph{minus} the components of the Riemann tensor in the conventions of~\cite{misner1974gravitation}.
\paragraph{The Berezin Measure and Stokes' Theorem.}
The Berezin measure on $\mathcal N$ becomes
\begin{equation}
    \Phi^*\left(|\det(g_{\mu\nu})|^{-\frac12} [dx\,d\theta\,d\tilde\theta]\right) = |\det(g_{\mu\nu})|^{\frac12} [dx\, d\theta\,d\chi] = \nu_g.
\end{equation}
Note that on $\mathcal N$, both $[dx\, d\theta]$ and $[d\theta\, d\chi]$ are invariant: $[dx\,d\theta]$ corresponds to the invariant measure on $T[1]M$, whereas $[d\theta \, d\chi]$ corresponds to duality pairing of the fiber directions of $\mathcal N$. We also have a version of Stokes' theorem on $\mathcal N$, which reads as follows:
\begin{theorem}
Let $\nu_g$ be the Berezin measure constructed above, let $a\in\mathcal O_\mathcal N$, and denote by
\begin{equation}
\mathcal N^\partial = T[1]\partial M \oplus T^*[1]M|_{\partial M}
\end{equation}
with induced Berezin measure
\begin{equation}
\nu^\partial_g = |\det(g)|^{\frac12}[dy\, d\psi\, d\chi],
\end{equation}
where $(y^\mu,\psi^\nu = \diff y^\nu,\chi_\lambda)$ are coordinates on $T[1]\partial M$, and $\chi_\lambda$ are the restricted degree 1 coordinates on $T^*[1]M|_{\partial M}$. It holds
\begin{equation}\label{eq:stokes-N}
\int_\mathcal N (Da)\nu_g = \int_{\mathcal N^\partial} a\,\nu_g^\partial.
\end{equation}
\end{theorem}
\begin{proof}
    The proof is completely analogous to that on $\mathcal M$ (see theorem~\eqref{thm:stokes-M}).
\end{proof}
\subsection{Graded Constrained Vielbein Gravity}
We now derive the constrained vielbein formalism from the graded Einstein-Cartan formulation of gravity.\par
For brevity, we consider pure gravity. Potential matter terms may be added just as in the metric formalism. Since vielbein fields are defined locally, we shall consider open coordinate patches of the graded manifold. In such an open coordinate patch, we may fix an arbitrary background metric $\eta_{\mu\nu}\theta^\mu\tilde\theta^\nu$. Let us first take the action from equation~\eqref{eq:graded-gravity}, pull it back to $\mathcal N$ (recall that the integral is invariant under diffeomorphisms), and evaluate it on $\eta$:
\begin{equation}
    S_\textrm{EC}(D,\eta) = \frac{1}{(n-2)!}\int_\mathcal U R_D \,\vartheta^{n-2}\, \nu_{\eta},
\end{equation}
where $R_D = \diff\varphi + \frac12\{\varphi,\varphi\}$, and the Poisson bracket is defined with respect to $\eta$.
In the above expression, $R_D$ and $\nu_{\eta}$ depend on $\eta$.
As we work in a local patch, we may choose the components $\eta_{\mu\nu}$ to be Minkowski, i.e., $\eta = \operatorname{diag}(-1,1\dots,1)$.

We now show how the action functional behaves under a change of metric. To that end, let $\tensor{e}{_\mu^\nu}$ be a $\eta$-symmetric invertible matrix, and write
\begin{equation}
    e=\tensor{e}{_\mu^\nu} \theta^\mu \chi_\nu \in\mathcal O_\mathcal N^{1,1}.
\end{equation} 
The condition that $e$ be $\eta$-symmetric is equivalent to $\{e,\vartheta\}=0$, or, equivalently, $\widetilde{e} = -e$. Now, any other metric $g$ may be written as
\begin{equation}
    g_{\mu\nu} = (e^* \eta)_{\mu\nu} = \tensor{e}{_\mu^\kappa}\eta_{\kappa\lambda} \tensor{e}{_\nu^\lambda}.
\end{equation}
By strictly considering non-degenerate $\eta$-symmetric $\tensor{e}{_\mu^\nu}$, we ensure that the datum of $\tensor{e}{_\mu^\nu}$ is equivalent to a finite deformation of the metric. We note here that $\tensor{e}{_\mu^\nu}$ behaves exactly like a vielbein field. The only difference is that classically, the upper index of $\tensor{e}{_\mu^\nu}$ is thought of as a Lorentz index rather than a spacetime index. However, since we contract this index with the locally chosen background metric $\eta_{\mu\nu}$, which we chose to be Minkowski, the distinction is merely notational.
A computation analogous to~\eqref{eq:action-finite-variation}
reveals that
\begin{align}
    S_\textrm{EC}(D,e^*\eta) &= \frac1{4(n-2)!}\int_\mathcal U \left(\tensor{e}{_\mu^\nu} \theta^\mu \chi_\nu \right)^{n-2}\tensor{R}{_\kappa_\lambda^\rho^\sigma}\theta^\kappa \theta^\lambda \chi_\rho \chi_\sigma  [dx\,d\theta\,d\chi]\\
    &= \frac{1}{(n-2)!} \int_\mathcal U e^{n-2} R_D,\label{eq:graded-PC}
\end{align}
where we used that $|\det(\eta)| = 1$, and chose to drop the flat Berezin measure $[dx\,d\theta\,d\chi]$ from notation. Thus, the resulting functional only depends on the choice of metric through the explicit appearance of $e = \tensor{e}{_\mu^\nu} \theta^\mu \chi_\nu$. This gives rise to the following interpretation: We consistently fix $\eta$ to be Minkowski in all coordinate patches, and view the functional as a functional of $e$. This prescription leads to the \emph{constrained vielbein action}
\begin{equation}
    S_\textrm{cv}(D,e) = \frac{1}{(n-2)!}\int e^{n-2} R_D.
\end{equation}
We add a generic matter functional $S_m(D,e,\dots)$, such that the full action is
\begin{equation}
    S(D,e,\dots) = S_\textrm{cv}(D,e) + S_m(D,e,\dots),
\end{equation}
and define
\begin{equation}
    \delta S_m = -\int \left(\mathcal T\,\delta e  + \mathcal S\,\delta\omega +\dots\right),
\end{equation}
where $\mathcal T$ is the vielbein stress-energy of $S_m$, $\mathcal S$ is the spin of $S_m$, and the remaining terms are due to the variation with respect to the matter fields. Without loss of generality, we take $\mathcal S = -\tilde{\mathcal S}$; the skew-symmetric contribution of $\tilde{\mathcal S}$ is not seen by the variation with respect to the symmetric $e$.\par
The variational principle for $S_\textrm{cv}$ is particularly simple as the action functional is polynomial in $e$. More specifically, we can use the graded Stokes' theorem to find
\begin{equation}
    \delta S_\textrm{cv} = \frac{1}{(n-3)!}\int \left( e^{n-3}\, R_D\, \delta e  + e^{n-3} De \,\delta\omega \right),
\end{equation}
assuming that a boundary condition for $\omega$ is provided. It follows that
\begin{equation}
    \frac{1}{2(n-3)!} e^{n-3} \left(R_D+\widetilde{R_D}\right)= \mathcal T,\quad \frac{1}{(n-3)!} e^{n-3} De = \mathcal S,
\end{equation}
where we used that only the symmetric part $\frac12(R_{D} + \widetilde{R_{D}})$ contributes as $e=-\tilde e$.
The field equations are equivalent to the vanishing of the Einstein tensor and the torsion of $D$. To relate this theory back to the metric formulation, recall that we view $g$ as a deformation of $\eta$ by $e$, that is,
\begin{equation}
    g_{\mu\nu}=(e^*\eta)_{\mu\nu} = \tensor{e}{_\mu^\kappa} \eta_{\kappa\lambda}\tensor{e}{_\nu^\lambda},
\end{equation}
and the corresponding Poisson vector field for $g$ may be recovered via the coefficients
\begin{equation}
    \tensor\varphi{_\mu_\nu^\kappa}=\tensor{e}{_\nu^\rho}\tensor{\omega}{_\mu_\rho^\sigma} \tensor{(e^{-1})}{_\sigma^\kappa} + \tensor{e}{_\nu^\rho}\partial_\mu \tensor{(e^{-1})}{_\rho^\kappa}
\end{equation}
in accordance with equation~\eqref{eq:d-extension-transformation}.
We note here that $\eta$ may be taken as an arbitrary background metric (rather than flat Minkowski in the local model), in which case equation~\eqref{eq:graded-PC} attains an additional factor of $|\det(\eta_{\mu\nu})|^{\frac12}$.

\subsection{Graded Palatini-Cartan Gravity}
We now discuss the graded PC theory of gravity. To that end, we simply take the constrained vielbein formulation and drop the symmetry constraint $\{e,\vartheta\}=0$. The resulting action reads
\begin{align}
    S_\textrm{PC}(e,D) &= \frac1{(n-2)!}\int_\mathcal U e^{n-2}\,R_D,\\ S(e,D,\dots)&=S_\textrm{PC}(e,D) + S_m(e,D,\dots).
\end{align}
To compute the field equations, we again impose
\begin{equation}
    \delta S_m = -\int_\mathcal U\left(\mathcal T \,\delta e + \mathcal S\,\delta\omega + \dots\right),
\end{equation}
where $\mathcal T$ and $\mathcal S$ are respectively stress-energy and spin of $S_m$, and the remaining terms are due to the variation with respect to the matter fields.
Since $e$ is unconstrained, we do not need to fix the antisymmetric contribution of $\mathcal T$. Varying the Palatini-Cartan action, we find
\begin{equation}
    \delta S_\mathrm{PC} = \frac{1}{(n-3)!}\int_\mathcal U\left( e^{n-3}\,R_D\,\delta e + e^{n-3}\, De\,\delta\omega\right),
\end{equation}
such that the field equations read
\begin{equation}
    \frac{1}{(n-3)!}e^{n-3}\, R_D = \mathcal T,\quad \frac{1}{(n-3)!}e^{n-3}\,De = \mathcal S,
\end{equation}
assuming that suitable boundary conditions are put on $\omega$. 
We note again that there is no ambiguity in $\mathcal T$ as $e$ is not constrained to be symmetric. In turn, $S_\textrm{PC}$ attains an additional gauge symmetry due to Lorentz transformations of the $\chi$ coordinates. In this way, our earlier constraint $e=-\tilde e$ may be interpreted as a gauge fixing of the vielbein field.

\subsection{The Minimal BV Action}
We now construct a minimal BV extension of the graded constrained vielbein theory. Throughout, we assume, for simplicity that $R_D$ is compactly supported (or decays sufficiently fast at infinity) and that $M$ has no boundary. We leave the boundary discussion for future work. Alternatively, one may work with fields satisfying Minkowski geometry $e=\vartheta = \theta^\mu\chi_\mu$ and $D=\diff = \theta^\mu\partial_\mu$ outside a compact submanifold $\mathcal C \subset \mathcal M$.

Note that since all constructions were carried out on $T^*[1]M\oplus T[1]M$ which is canonically associated with $M$, the action functional is invariant under $\operatorname{Diff}(M)$, which is infinitesimally generated by $\mathfrak{X}(M)$ via vector fields, acting by Lie derivatives on $\mathcal O_\mathcal N$. Since we fixed $e$ to be $\eta$-symmetric, there is no residual Lorentz symmetry that needs to be dealt with, unlike in ~\cite{canepa2021general,grigoriev2021presymplectic}. 

For transparent bookkeeping, we label the \emph{ghost degree} of a field $\psi$ by $\gh(\psi)$, and write $A[j_\gh]$ whenever we shift the ghost degree of some object $A$ by $j$. In particular, by $\mathcal O_\mathcal N[j_\gh]$, we mean graded functions of the form
\begin{equation}
    a = \sum_{M,N} \tensor{a}{_M^N} \theta^M \chi_N,
\end{equation}
where the coefficients $\tensor{a}{_M^N}$ are \emph{declared} to have ghost degree $j$.
Note that the ghost degree is captured by the coefficients $\tensor{a}{_M^N}$ rather than the generators $\theta,\chi$.

Let us indicate vector fields on $M$, viewed as differential operators, with a hat ($\hat{-}$) symbol.
The full degree $+1_\gh$ BV differential is induced by some vector field $\hat c = c^\mu\partial_\mu\in\mathfrak X (M)[1_\gh]$, and reads
\begin{equation}
    Q\coloneqq L_{\hat c} \in\mathfrak X(\mathcal O_\mathcal N)[1_\gh]
\end{equation}
on the physical fields. Here, $L$ denotes the Lie derivative, $\gh(c^\mu) = 1$, and $c^\mu$ is of degree $0$ in $\theta$ and $\chi$.
Now, the BV differential $Q$ acts as
\begin{equation}
    Q(x^\mu) = c^\mu,\quad Q(\theta^\mu) = \theta^\nu \partial_\nu c^\mu,\quad Q(\chi_\nu) = -(\partial_\nu c^\mu) \chi_\mu
\end{equation}
on the generators of $\mathcal O_\mathcal N$.
With this, the action of $Q$ on the physical field $e$ may be computed as follows:
\begin{align}
    Q(e) &= c^\mu (\partial_\mu\tensor{e}{_\nu^\kappa})\theta^\nu\chi_\kappa +\tensor{e}{_\nu^\kappa}\theta^\nu (\partial_\kappa c^\mu) \chi_\mu + \tensor{e}{_\nu^\kappa} \theta^\mu(\partial_\mu c^\nu)\chi_\kappa,
\end{align}
and for the coefficients $\tensor{\omega}{_\mu_\kappa^\lambda}$ of the Poisson derivation, we find
\begin{equation}
      g^{\kappa\sigma}Q(\omega)\tensor{}{_\mu_\sigma^\lambda}  = c^\rho \tensor{R}{_\rho_\mu^\kappa^\lambda} + g^{\kappa\sigma}\nabla_\mu\nabla_\sigma c^\lambda.
\end{equation}
Since we only need the antisymmetric component in $\kappa,\lambda$ of the above expression in the BV action, we may write
\begin{equation}
    Q(\omega) = \frac12 \left(c^\rho \tensor{R}{_\rho_\mu^\kappa^\lambda} + g^{\kappa\sigma}\nabla_\mu\nabla_\sigma c^\lambda\right) \theta^\mu \chi_\kappa\chi_\lambda \in\mathcal O^{1,2}_\mathcal N [1_\gh].
\end{equation}
We now compute $\widehat{Q(c)}$ for the non-physical ghost field $\hat c$.
To that end, let $\phi$ be some physical field of the theory. Imposing $Q^2 = 0$, together with the Leibniz rule, we find
\begin{equation}
    0 = Q(Q(\phi)) = Q(L_{\hat c}\phi) = L_{\widehat{Q( c)}}\phi -\frac12L_{[\hat c,\hat c]}\phi,
\end{equation}
where we used that $L_{\hat c}^2 = \frac{1}{2} L_{[\hat c,\hat c]}$ as $c$ is odd. Thus, we find $\widehat{Q( c)} = -\frac12[\hat c,\hat c]$, or, in components,
\begin{equation}
    \widehat{Q(c)} = -c^\mu(\partial_\mu c^\nu)\partial_\nu \in\mathfrak X(M)[2_\gh].
\end{equation}
We may represent $\hat c$ and $\widehat{Q(c)}$ as functions $c$ and $Q(c)$ on $\mathcal N$ via
\begin{equation}
    c = c^\mu \chi_\mu\in \mathcal O^{0,1}_\mathcal N[1_\gh],
    \quad
    Q(c) = -c^\mu(\partial_\mu c^\nu)\chi_\nu\in\mathcal O_\mathcal N[2_\gh],
\end{equation}
allowing for a particularly concise graded geometric setup of the BV theory.
With this, the space of fields $\mathcal F$ of the BRST extension of the constrained graded vielbein theory is affine over $ \operatorname{Sym}(\mathcal O_\mathcal N^{1,1})\times\mathcal O^{1,2}_\mathcal N \times \mathcal O_\mathcal N^{0,1}[1_\gh]$,\footnote{Recall that while $\omega$ is not a function on $\mathcal N$, the difference $\delta \omega =\omega-\omega'$ is a Hamiltonian vector field $\{h,-\}$ for some $h\in\mathcal O^{1,2}_\mathcal N$.} where $\operatorname{Sym}(\mathcal O_\mathcal N^{1,1}) = \{e\in\mathcal O_\mathcal N^{1,1}:\tilde e = -e\}$. Moreover, $Q(\omega)\in\mathcal O^{1,2}_\mathcal N [1_\gh]$ and $R_D\in\mathcal O_\mathcal N^{2,2}$ are functions on $\mathcal N$.

In conclusion, all fields occurring in the action functional are functions on $\mathcal N$ (potentially with additional ghost degree). In the usual BV spirit, the antifields of $e,\omega$ and $c$ are elements of the BV extended field space which is the shifted cotangent bundle $T^*[-1_\gh]\mathcal F$. 
In our graded formalism, we have a particularly simple characterisation of the BV extended field space: Note that we have a canonical pairing
\begin{equation}
    \mathcal O^{p,q}_\mathcal N \times \mathcal O^{n-p,n-q}_\mathcal N \xrightarrow{\int_\mathcal N \nu} \mathbb R,
\end{equation}
with $\nu$ the integral density associated with the background metric. This identifies
\begin{equation}
    T^*[-1_\gh]\mathcal O^{p,q}_\mathcal N \cong \mathcal O^{p,q}_\mathcal N\times\left(\mathcal O^{p,q}_\mathcal N \right)^*[-1_\gh] \cong \mathcal O^{p,q}_\mathcal N\times\mathcal O^{n-p,n-q}_\mathcal N[-1_\gh],
\end{equation}
such that $T^*[-1_\gh]\mathcal O_\mathcal N \cong \mathcal O_\mathcal N \times\mathcal O_\mathcal N [-1_\gh]$. Here, the cotangent bundle is understood as the $L^2$ dual of the tangent bundle. We conclude that we do not need further graded variables to account for the antifields.

We now introduce the corresponding antifields $e^\dagger,\omega^\dagger$ of ghost degrees $-1_\gh$ and $\hat c^\dagger$ of ghost degree $-2_\gh$:
\begin{equation}
    e^\dagger \in \mathcal O^{n-1,n-1}_\mathcal N[-1_\gh],\quad  \omega^\dagger \in \mathcal O^{n-1,n-2}_\mathcal N[-1_\gh], \quad c^\dagger \in \mathcal O_\mathcal N^{0,1}[-2_\gh].
\end{equation}
As pointed out above, the antifields are again functions on $\mathcal N$, endowed with appropriate ghost degrees.

Let us again work locally as in the preceding section, and choose the background metric to be Minkowski with respect to the chosen coordinates. We suppress the flat measure $[dx\,d\theta\,d\chi]$ from notation.
The symplectic form on the space of fields is given by
\begin{equation}
    \varpi = \int_\mathcal N \left(\delta e^\dagger\wedge\delta e + \delta\omega^\dagger \wedge\delta\omega + \delta c^\dagger \wedge\delta c\right),
\end{equation}
where $\wedge$ denotes the product of functions on $T[1]\mathcal O_\mathcal N$, and $\delta$ is the homological vector field on $T[1]\mathcal O_\mathcal N$, or, equivalently, the de Rham differential on the space of fields.
The full BV action reads
\begin{equation}
    S_\mathrm{BV} = \int_\mathcal N \left( \frac{1}{(n-2)!}e^{n-2} \, R_D + e^\dagger Q(e) + \omega^\dagger Q(\omega) + c^\dagger Q(c)\right).
\end{equation}
The symplectic form on field space induces the odd Poisson structure
\begin{equation}
    (A,B) = \frac12\int_\mathcal N \sum_{\phi\in \{e,\omega,c\}}
    \left( \frac{\delta A}{\delta \phi^\dagger}\frac{\delta B}{\delta\phi} + (-1)^{|A||B|}\frac{\delta B}{\delta \phi^\dagger}\frac{\delta A}{\delta\phi}\right)
\end{equation}
of ghost degree $+1_\gh$,
and by plugging in the BV action, we find
\begin{equation}\label{eq:cme}
    (S_\textrm{BV},S_\textrm{BV}) = \frac{1}{(n-3)!}\int_\mathcal N\left(Q(e) e^{n-3} R_D + Q(\varphi)e^{n-3}De \right).
\end{equation}
All other terms vanish by Stokes' theorem $\int Q(-) = 0$, together with $Q(A)Q(B) =Q(A Q(B))$ as $Q^2 = 0$.
Equation~\ref{eq:cme} vanishes as it corresponds precisely to the infinitesimal gauge transformation by the diffeomorphism generated by $c^\mu \partial_\mu$. More concretely, an explicit calculation shows that
\begin{equation}
    \frac{1}{(n-3)!}\int_\mathcal N\left(Q(e) e^{n-3} R_D + Q(\omega)e^{n-3}De \right)[dx\,d\theta\,d\chi] = \frac2{(n-2)!}\int_\mathcal N Q\left(e^{n-2} R_D [dx\,d\theta\,d\chi]\right),
\end{equation}
which vanishes by Stokes' theorem whenever $\mathcal M$ has no boundary and $R_D$ has compact support.\par
To obtain the minimal BV extension of the graded Palatini-Cartan theory, one needs to take care of the additional Lorentz symmetry $\chi_\mu \to \tensor{\Lambda}{_\mu^\nu} \chi_\nu$. Recall from equation~\eqref{eq:lorentz-vielbein} that such transformations are infinitesimally generated by Hamiltonian vector fields of the form $\{\lambda,-\}$ for $\lambda\in\mathcal O_\mathcal N^{0,2}[1_\gh]$. The incorporation of these Lorentz ghosts is known; see, e.g.,~\cite{canepa2021general,grigoriev2021presymplectic} and references therein.

\section{Conclusion and Outlook}
In this paper, we showed how to graded geometric structures introduced in~\cite{chatzistavrakidis2017tensor, chatzistavrakidis2020unified, chatzistavrakidis2020graded} to the case where the spacetime manifold is curved.
We then used these tools to cast the Einstein-Cartan and Palatini-Cartan theories of gravity in the graded geometric formalism. Using a diffeomorphism of graded manifolds, we showed that these theories are equivalent if one restricts to vielbein fields that are symmetric with respect to a background metric. As such vielbein fields always exist and reproduce the correct physical metric, this does not pose a restriction on the physical applicability of the theory. By removing the constraint on the vielbein field, we demonstrated how to obtain the Palatini-Cartan theory.

After treating the classical theory, we showed how to deal with the diffeomorphism symmetry by constructing the minimal BV-extended action.\par

As the theories are naturally expressed on the graded manifolds $T[1]M\oplus T[1]M$ and $T[1]M\oplus T^*[1]M$, a natural question is how generalized geometry fits into the picture. To that end, it would be desirable to establish a link to graded Courant algebroids~\cite{severa2020courant,chatzistavrakidis2026gravity}. Additionally, a graded BFV theory needs to be developed to deal with nonempty boundaries~\cite{canepa2021general}. It might also be useful to compare our current covariant graded geometric study to jet formulations involving the graded manifold $T[1]T[1]M$~\cite{kochan2003differential}.

Finally, it would be of great value to extend the setup to the case where $M$ is already a $\mathbb Z_2$-graded supermanifold, allowing for a direct extension to superspace formulations of supersymmetric field theories.

\acknowledgments
The authors thank Leonardo Castellani, Athanasios Chatzistavrakidis, Henry Ruben Fischer, Larisa Jonke, Maxim Grigoriev, Rishabh Jha, Tay Pang Shiang, Richard Szabo and Jan Vysok\'y for insightful conversations.
This work is funded by the Deutsche Forschungsgemeinschaft (DFG, German Research Foundation) – grant no.~545733782.


\bibliographystyle{JHEP}
\bibliography{sources.bib}
\appendix
\section{Berezin Integrals}\label{sec:app-berezin}
We now develop a covariant formulation of Berezin integrals on $\mathcal U \cong U \times\mathbb R^s[k]$ for odd $k$ and contractible $U$. We equip $\mathcal U$ with global coordinates $(x^\mu,\xi^\mu)$, where $x^\mu$ are in degree 0, and the $\xi^i$ are in degree $k$. Any function $a$ on $\mathcal U$ can then be decomposed according to
\begin{equation}
    a = \sum_I a_I \xi^I.
\end{equation}
The \emph{Berezin integral} of $F$ is defined as
\begin{equation}
    \int_\mathcal U a [d x^1\dots d x^n \,d\xi^1\dots d\xi^s] \coloneqq \int_U a_{1,\dots,1} d\lambda( x^1,\dots, x^n)
\end{equation}
with respect to the Lebesgue measure. In words, we \emph{take the coefficient of the highest degree monomial and integrate it with respect to the Lebesgue measure}. We thus have the alternative description
\begin{equation}\label{eq:berezin-pdiff}
    \int_{\mathcal U} a [d x^1\dots d x^n d\xi^1\dots d\xi^s] = \int_U \iota^*\left(\pdiff{}{\xi^1} \dots \pdiff{}{\xi^s} a \right)d\lambda(x^1,\dots,x^n),
\end{equation}
where $\iota:U\to U\times\mathbb R^s [k]$ is the inclusion map.
In practice, $\iota^*$ is implied and dropped from the notation. This translates to the ``recipe"
\begin{equation}
    \int \xi^1\dots\xi^s [\diff\xi^1 \dots \diff\xi^s] = 1,
\end{equation}
together with Lebesgue integration along the degree 0 coordinates.
If we do not have any negatively graded coordinates, any degree-preserving diffeomorphism $U\times\mathbb R^s[k] \to W\times \mathbb R^s [k]$ automatically has the form\footnote{In the case where both negative and positive grading are present, the following results still hold, but the computations generally become more complicated (essentially because the Berezinian also ``sees'' the off-diagonal blocks, and $\phi$ generally allows for maps of the type $\phi^*(x^\mu) = x^\mu + \theta^i c_{ij} \xi^j$ with $\deg(\xi) = -\deg(\theta)$.}
\begin{equation}\label{eq:diffeo}
\phi: \mathcal U \to \mathcal W,\quad \phi^*(y^i) = x^i(y), \quad\phi^*(\eta^k) = \xi^k(y,\eta).
\end{equation}
It follows that the Berezinian of the Jacobian decomposes into block-trigonal form:
\begin{equation}
    \operatorname{ber}(J\phi) = \operatorname{ber}\begin{pmatrix}
        \pdiff{x}{y} & 0\\
        \pdiff{\xi}{y} & \pdiff{\xi}{\eta}
    \end{pmatrix}
    =\det\left(\pdiff{\eta}{\xi}\right) \det\left(\pdiff{x}{y}\right).
\end{equation}
This allows us to write down the following transformation formula for $a\in\mathcal O_\mathcal W$:
\begin{align}
    \int_\mathcal U (a\circ\phi) [\diff x \,\diff\xi] &= \int_U\pdiff{}{\xi^1} \dots \pdiff{}{\xi^s} (a\circ\phi) d\lambda(x) \\
    &= \int_U \det\left(\pdiff{\eta}{\xi}\right)\left(\pdiff{}{\eta^1} \dots \pdiff{}{\eta^s} a\right)\circ \phi\, d\lambda(x) \\
    &= \int_W \det\left(\pdiff{\eta}{\xi}\right) \left|\det\left(\pdiff{x}{y}\right)\right|a \,d\lambda(y)\\
    &=\int_\mathcal W \operatorname{or}(\phi)\operatorname{ber}(D\phi) a [\diff y\,\diff\eta],
\end{align}
where $\operatorname{or}(\phi) = \operatorname{sign}\det(\pdiff{x}{y})$, and $[\diff x\,\diff\xi] = [\diff x^1\dots\diff x^n\diff\xi^1\dots\diff\xi^s]$. The integral operator thus becomes invariant under diffeomorphisms if we integrate objects that transform as
\begin{equation}
    \phi^*\beta = \operatorname{or(\phi)}\operatorname{ber}(D\phi) \beta \circ\phi.
\end{equation}
The $\mathcal O (\mathcal U)$-module corresponding to such objects is called \emph{sheaf of integral densities}, and is 
\begin{equation}
	\Gamma(\operatorname{or}(T\mathcal U)\otimes_{\mathbb Z_2}\operatorname{ber}(T\mathcal U)) \eqqcolon \mathcal I(\mathcal U).
\end{equation}
If $\mathcal M$ is oriented and only orientation preserving diffeomorphisms are considered, sections of the Berezinian sheaf $\mathcal Ber(\mathcal U) = \varphi(\operatorname{ber}(T\mathcal U))$ can be integrated directly.
\section{Stokes' Theorem}
Let us first formulate the general Stokes' theorem on graded manifolds without boundary. Recall that each vector field $X$ on a graded manifold $\mathcal M$ induces a flow $\phi^X_t$, and sections of natural vector bundles may be differentiated in the direction of $X$ in the following sense
\begin{equation}
	L_X(\mu) = \left.\frac\diff{\diff t}\right\vert_{t=0} (\phi_t^X)^*(\mu),
\end{equation}
where $t$ is an auxiliary variable of degree $-\deg(X)$. In the following, we will consider the case where $\mu\in\mathcal I(\mathcal M)$ is a Berezin measure.
\begin{theorem}[Stokes on graded manifolds] Let $\mathcal M$ be a graded manifold without boundary, let $X$ be a vector field on $\mathcal M$, and let $\mu\in\mathcal I_\mathcal M$ be a compactly supported integral density. It holds
\begin{equation}
\int_\mathcal M L_X(\mu) = 0.
\end{equation}
\end{theorem}
\begin{proof}
Along the even directions, this follows from the ordinary Stokes' theorem; for the odd directions, one sees that expressions of the kind $\pdiff{f}{\xi^i}$ do not contain any top power in $\xi^i$ and therefore, the Berezin integral vanishes. For completeness, let us give a more invariant proof of the theorem. For $\mathcal M$ without boundary, we have
\begin{equation}
	\int_\mathcal M L_X\mu = \left.\frac\diff{\diff t}\right\vert_{t=0} \int_\mathcal M (\phi^X_t)^*\mu = 0
\end{equation}
by diffeomorphism invariance.
\end{proof}
Defining graded manifolds with boundaries is generally not straight-forward. For the purposes of this paper, it suffices however to consider special cases that we know how to deal with. In what follows, we derive two versions of Stokes' theorem which apply to the graded geometric language presented in the main text. To do so, we make use the classical Stokes' theorem for ordinary manifolds.
\begin{lemma}
	Let $F\in \mathcal O_{T[1]M}$, and denote by $\diff$ the homological vector field. It holds
	\begin{equation}
		\int_{T[1]M} \diff a = \int_{T[1]\partial M}a,
	\end{equation}
	where both integrals are computed with respect to the canonical invariant measures $[dx\, d\theta]$ on $T[1]M$ resp.~$T[1]\partial M$. 
\end{lemma}
\begin{proof}
We explicitly write out the invariant measure $[dx\,d\theta]$ and carry out the $\theta^\mu$ integration first, such that
\begin{equation}
\int_{T[1]M} (\diff a)_N \theta^N [dx\, d\theta] = \int_M (\diff a)_N \diff x^N = \int_{\partial M} a_L \diff y^L = \int_{T[1]\partial M} a,
\end{equation}
where we identified the cohomological vector field $\diff$ on $T[1]M$ with the exterior derivative of differential forms on $M$. The middle equality holds due to Stokes' theorem on ordinary (non-graded) manifolds.
\end{proof}
To extend this result to $\mathcal E = T[1]M\oplus E[k]$, we need to extend the Berezin measure to $\mu = [dx\,d\theta]\mu_E$, where $\mu_E = \rho [d\xi]$. We assume further that $\rho$ is homogeneous of degree $0$. Let us denote by $\iota:\partial M\xhookrightarrow{} M$ the inclusion of the boundary, and define 
\begin{equation}
	\mathcal E^\partial\coloneqq T[1]\partial M \times_{\partial M} (\iota^*E)[k] = T[1]\partial M \oplus E|_{\partial M}[1].
\end{equation}
Note that $\mu$ induces a measure $\mu^\partial$ on $\mathcal E^\partial$ by
\begin{equation}
\mu^\partial = \rho|_{\partial M} [dy\,d\eta\,d\xi],
\end{equation}
where $(y^\mu,\chi^\nu = \diff y^\nu)$ are coordinates on $T[1]\partial M$.
\begin{theorem}\label{thm:stokes-graded-general}
Let $\mathcal E = T[1]M\oplus E[k]$, and let $\mu = \rho[dx\,d\theta\,d\xi]$ be a Berezin measure on $\mathcal E$ such that $\deg(\rho)=0$ and $L_D\mu = 0$. It holds
\begin{equation}
\int_\mathcal E (Da)\mu = \int_{\mathcal E^\partial} a \mu^\partial.
\end{equation}
\end{theorem}
\begin{proof}
For convenience, we may choose local coordinates $(x^\mu)$ of degree $0$ in a neighbourhood of the boundary such that $x^n = 0$ on $\partial M$. In other words, we have $x^\mu = y^\mu$ for $\mu<n$. We now compute
\begin{align}
\int_\mathcal E (Da)\,\mu &= \int_{T[1]M} \pdiff{}{\xi^s}\dots\pdiff{}{\xi^1} (Da) \rho[dx\,d\theta]\\
&=\int_{T[1]M} \diff\left( \rho \pdiff{}{\xi^s}\dots\pdiff{}{\xi^1} a\right) [dx\,d\theta]\\
&=\int_{T[1]\partial M} \rho \pdiff{}{\xi^s}\dots\pdiff{}{\xi^1} (a) [dy\,d\chi]\\
&=\int_{\mathcal E^\partial} a\,\mu^\partial
\end{align}
where the second equality holds since $L_D(\mu) = 0$ and $D = \diff$ once all the $\xi$ coordinates have been integrated out.
\end{proof}
In the cases that we are interested in, $\mathcal E$ will always come with a Berezin measure $\mu = \mu_E \cdot [dx\, d\theta] = \rho [dx\,d\theta\,d\xi]$, such that the extension $D$ of $\diff$ satisfies $L_D(\mu) = 0$.

\section{\texorpdfstring{Invertibility of the $\star$ Operator}{Invertibility of the Star Operator}}\label{app:properties-star}
In this section, we prove some properties of the $\star$ operator. Without loss of generality, we will work on $\mathcal M$; the results carry over to $\mathcal N$ upon identifying $\mathcal M\cong \mathcal N$ via $\tilde\theta^\mu\mapsto g^{\mu\nu}\chi_\nu$. 
Recall the definition of the $\star$ operator on $\mathcal M$,
\begin{equation}
    \star f = \begin{cases}
        \frac1{(n-p-q)!} g^{n-p-q}\tilde a,&p+q\le n,\\
        \frac{1}{(p+q-n)!}\trg^{p+q-n}\tilde a, &p+q>n,
    \end{cases}
\end{equation}
where $a\in\mathcal O^{p,q}_\mathcal M$. The following lemmas will lay the foundation to prove the following
\begin{theorem}\label{thm:star-inv}
    The $\star$ operator is an isomorphism $\star:\mathcal O^{p,q}_\mathcal M \xrightarrow{\sim} \mathcal O^{n-p,n-q}_\mathcal M$. Its eigenvalues in $\mathcal O_\mathcal M\otimes\mathbb C$ are Gaussian integers, that is, $\sigma(\star)\subset\mathbb Z[i]$.
\end{theorem}
We break up this problem into several intermediate results.
In particular, it turns out that this problem may be reduced to properties of representations of $\mathfrak{sl}_2$. Recall to that end that over any field of characteristic $\neq 2$, $\mathfrak{sl}_2$ is generated by $E,F,H$, subject to the relations
\begin{equation}
    [E,F]=H,\quad [H,E]= 2E,\quad [H,F] = -2F.
\end{equation}
More concretely, one may identify
\begin{equation}
    E =\begin{pmatrix}
        0&1\\0&0
    \end{pmatrix},\quad
    F =\begin{pmatrix}
        0&0\\1&0
    \end{pmatrix},\quad
    H =\begin{pmatrix}
        1&0\\0&-1
    \end{pmatrix}.
\end{equation}
\begin{lemma}\label{lemma:sl2-rep}
    Let $\dim(M) = n$.
    The algebra of operators $g,\trg,\deg - n$ on $\mathcal M$ satisfies $\mathfrak{sl}_2$ commutation relations, that is,
    \begin{equation}
        [g,\trg] = \deg - n,\quad [\deg-n, g] = 2g,\quad [\deg-n,\trg] = -2\trg.
    \end{equation}
    In other words, $\rho:\mathfrak{sl}_2 \to \operatorname{DiffOp}(\mathcal M)$, defined by
    \begin{equation}
        \rho(E) = g,\quad \rho(F)=\trg,\quad \rho(H) = \deg - n,
    \end{equation}
    is a Lie algebra representation.
\end{lemma}
\begin{proof}
    The first commutation relation may be verified by means of a direct computation. For the second and third relations, note that $g$ increases the degree by $2$, whereas $\trg$ lowers the degree by $2$.
\end{proof}

\begin{lemma}\label{lemma:lefschetz-decomp}
    Let $a\in\mathcal O^{p,q}_\mathcal M$ with $p+q< n$. There exists a unique decomposition $a = a_0 + g a_1$, where $\trg a_0 = 0$. Equivalently, \begin{equation}
        \mathcal O_\mathcal M^{p,q} = \left(\mathcal O^{p,q}_\mathcal M\right)_0 \oplus g\mathcal O_\mathcal M^{p-1,q-1},
    \end{equation}
    where $\left(O^{p,q}_\mathcal M\right)_0 = \{b\in\mathcal O^{p,q}_\mathcal M:\trg b = 0\}$ denotes the traceless part.
    Whenever $p+q>n$, we analogously have the decomposition
    \begin{equation}
        \mathcal O_\mathcal M^{p,q} = \left(\mathcal O^{p,q}_\mathcal M\right)_0 \oplus \trg \mathcal O_\mathcal M^{p+1,q+1},
    \end{equation}
    where $\left(\mathcal O^{p,q}_\mathcal M\right)_0 = \{b\in\mathcal O^{p,q}_\mathcal M : gb = 0\}$ denotes the part with vanishing cotrace.
\end{lemma}
\begin{proof}
    We prove the first decomposition by induction on $\min(p,q)$. Without loss of generality, let $p\le q$. The base case $p=0$ is clear as all $a\in\mathcal O^{0,q}_\mathcal M$ satisfy $\trg a = 0$, and $\mathcal O^{-1,q-1}_\mathcal M = \{0\}$. For the induction step, note that $\mathcal O^{p+1,q}_\mathcal M$ is generated by functions of the form $bc$, where $b\in\mathcal O^{1,0}_\mathcal M$ and $c\in\mathcal O_\mathcal M^{p,q}$. It obviously holds $\trg b = 0$. Using the induction hypothesis, we have
    \begin{equation}
        bc = b(c_0 + g c_1) = bc_0 + gbc_1.
    \end{equation}
    It thus remains to find the decomposition $bc_0 = (bc_0)_0 + g(bc_0)_1$. To that end, note that
    \begin{equation}
        \trg bc_0 = -\pdiff{}{b} c_0,
    \end{equation}
    where $\pdiff{}{b} = g^{\mu\nu} b_\mu\pdiff{}{\theta^\nu}$. Note further that $[\trg,\pdiff{}{b}] = 0$. Furthermore, we have
    \begin{equation}
        (p+q+1-n) \trg (b c_0) = [\trg, g] bc_0 = \trg(gbc_0) - g\trg\pdiff{}{b}c_0 = \trg(gbc_0),
    \end{equation}
    such that $(bc_0)_0 = bc_0 + \frac{1}{n-p-q-1}gbc_0$ and $(bc_0)_1 = \frac{-1}{n-p-q-1}bc_0$.\par
    To prove the second result, proceed by downwards induction on $\max(p,q)$. Clearly, if $\max(p,q) = n$, it holds $ga = 0$ for all $a\in\mathcal O^{p,q}_\mathcal M$, proving the base case. Without loss of generality, assume $p>q$. For the induction step, note that $\mathcal O^{p-1,q}_\mathcal M$ consists of linear combinations of the type $b^\mu\pdiff{}{\theta^\mu} c$ for $c\in\mathcal O^{p,q}_\mathcal M$. For the full induction, proceed as above. 
\end{proof}
\begin{corollary}
    For $p+q<n$, we have the decomposition
    \begin{equation}
        \mathcal O^{p,q}_\mathcal M = \bigoplus_{k=0}^{\min(p,q)} g^k \left(\mathcal O^{p-k,q-k}_\mathcal M\right)_0.
    \end{equation}
    For $p+q>n$, we have the decomposition
    \begin{equation}
        \mathcal O^{p,q}_\mathcal M = \bigoplus_{k=0}^{n-\max(p,q)} \trg^k \left(\mathcal O^{p+k,q+k}_\mathcal M\right)_0.
    \end{equation}
\end{corollary}
\begin{proof}
    This follows by inductively applying lemma~\eqref{lemma:lefschetz-decomp}.
\end{proof}
We prove invertibility of $\star$ by showing that $\star^2$ is invertible. To that end, we require the commutation relations of $g^k$ and $\trg^\ell$. To evaluate these, we may go to the universal enveloping algebra $U(\mathfrak{sl}_2)$ and use the following
\begin{lemma}\label{lemma:straightening}
    Let $E,F,H$ be the generators of $\mathfrak{sl_2}$ as before. In the universal envelope $U(\mathfrak{sl}_2)$, it holds
    \begin{equation}
        E^k F^\ell = \sum_{j=0}^{\min (k,\ell)} \frac{\ell!}{(\ell - j)!}F^{\ell-j} \binom{H-k-\ell+2j}{j}\frac{k!}{(k-j)!}E^{k-j},
    \end{equation}
    as well as
    \begin{equation}
        F^k E^\ell = \sum_{j=0}^{\min (k,\ell)} \frac{\ell!}{(\ell - j)!}E^{\ell-j} \binom{-H-k-\ell+2j}{j}\frac{k!}{(k-j)!}F^{k-j}.
    \end{equation}
\end{lemma}
\begin{proof}
    The first formula may be proven by induction on $k$ and $\ell$, together with standard identities for binomial coefficients. The details are spelled out in~\cite{humphreys1972representation}.\par
    The second formula follows from the first after applying the Lie algebra automorphism
    \begin{equation}
    \kappa:\mathfrak{sl}_2\to\mathfrak{sl}_2,\quad
        \kappa(E) = F,\quad \kappa(F) = E,\quad \kappa(H) = -H.
    \end{equation}
\end{proof}
By using the decomposition lemma~\eqref{lemma:lefschetz-decomp}, it suffices to compute $\star^2$ on functions of the kind $g^k a$ with $\trg a = 0$, and functions of the kind $\trg^k a$ with $ga=0$. 

\begin{lemma}
    For $a = g^m b \in\mathcal O_\mathcal M^{p,q}$ with $p+q< n$ and $\trg b=0$, it holds
    \begin{equation}\label{eq:star2-1}
        \star^2 a = (-1)^{m+n-p-q}\binom{m+n-p-q}{n-p-q}^2 a.
    \end{equation}
    For $a = \trg^m b \in\mathcal O_\mathcal M^{p,q}$ with $p+q> n$ and $gb=0$, it holds
    \begin{equation}
        \star^2 a = (-1)^{m+p+q-n}\binom{m+p+q-n}{p+q-n}^2 a.
    \end{equation}
    For $a\in\mathcal O^{p,q}_\mathcal M$ with $p+q=n$, it holds $\star^2a = a$.    
\end{lemma}
\begin{proof}
    We prove the first case. With $k=n-p-q$, it holds
    \begin{equation}
        \star^2(g^m b) = \frac{(-1)^{k+m}}{(k!)^2}\trg^k (g^{k+m} b).
    \end{equation}
    We now use lemma~\eqref{lemma:straightening} with $g=\rho(E)$, $\trg = \rho(F)$, and $\ell = k+j$, and since $\trg b = 0$, only the term with $k=j$ contributes. Collecting all the combinatorial factors leads to equation~\eqref{eq:star2-1}.\par   
    The second case is analogous. For the third case, observe that $\star a = \tilde a$ whenever $\deg(a) = n$.
\end{proof}
By lemma~\eqref{lemma:lefschetz-decomp}, all functions on $\mathcal M$ can be written as linear combinations of such functions. This also shows that these functions are are precisely the eigenfunctions of $\star^2$. Since the eigenvalues of $\star^2$ are non-zero perfect squares in $\mathbb Z[i]$, this proves theorem~\eqref{thm:star-inv}.

\end{document}